\documentclass[11pt]{article}
\usepackage[margin=1in]{geometry}
\usepackage{cite}
\usepackage{amsmath,amssymb,amsfonts,amsthm}
\usepackage{algorithmic,url,tikz}
\usepackage{graphicx,booktabs,subcaption}
\usepackage{textcomp}
\usepackage[dvipsnames]{xcolor}
\def\BibTeX{{\rm B\kern-.05em{\sc i\kern-.025em b}\kern-.08em
    T\kern-.1667em\lower.7ex\hbox{E}\kern-.125emX}}

\allowdisplaybreaks

\newtheorem{thm}{Theorem}
\newtheorem{lem}[thm]{Lemma}

\theoremstyle{definition}
\newtheorem{defn}[thm]{Definition}

\theoremstyle{remark}

\newcommand{\cF}{{\mathcal{F}}}

\newcommand{\cS}{{\mathcal{S}}}

\renewcommand{\phi}{\varphi}

\renewcommand{\le}{<}
\renewcommand{\ge}{>}

\newcommand{\lrp}[1]{\left({#1}\right)}

\newcommand{\lrs}[1]{\left\{{#1}\right\}}
\newcommand{\lrm}[1]{\left|{#1}\right|}

\newcommand{\OPT}{\operatorname{OPT}}
\newcommand{\cost}{\operatorname{cost}}

\title{Pre-emptive Carbon Aware Scheduling for Edge Workloads\footnote{Supported by NSF 2216970 (IDEAL) as well as a gift from Adobe Labs.}}
\date{}

\author{Cecilia Chen\thanks{\texttt{zhuo.chen@northwestern.edu}} \and Rajni Dabas \thanks{\texttt{rajni.dabas1@northwestern.edu}} \and 
\and Samir Khuller\thanks{\texttt{samir.khuller@northwestern.edu}} \and Emilie Rivkin\thanks{\texttt{emilie.rivkin@northwestern.edu}}}

\begin{document}
\maketitle
\begin{abstract}
 Power grids use a constantly changing mix of fossil fuels and clean renewables like solar and wind. Because of this, the {\em carbon footprint} fluctuates drastically based on time and geographic location. Time-varying carbon intensity creates opportunities to reduce carbon emissions by shifting workloads to cleaner periods, but deadlines, task dependencies, and limited edge capacity constrain this flexibility. We study carbon-aware scheduling of jobs with release times, deadlines, and arbitrary precedence constraints under piecewise-constant carbon costs. Jobs execute preemptively on a single edge server or are outsourced at job-dependent costs to an unlimited-capacity cloud, modeled as providing instantaneous execution. The objective is to minimize the sum of edge carbon costs and cloud outsourcing costs while satisfying all timing and precedence constraints. For edge-only execution, we give an exact polynomial-time algorithm. Allowing cloud outsourcing makes the problem NP-hard, but we obtain a fully polynomial-time approximation scheme through a reduction to preemptive scheduling with rejection penalties (weighted tardiness). In contrast, even two edge servers without job migration admit no polynomially bounded approximation when carbon-cost variation is unrestricted, unless \(P=NP\). We complement these theoretical results with experiments on controlled synthetic workloads and production-derived Alibaba traces. On synthetic instances, our incremental max-flow implementation is $4.96\times$ faster on average than a direct min-cost flow, while experiements on production-derived workloads evaluate its performance under realistic job structures.
 
 % or the second sentence can be: Our evaluation compares the proposed algorithms with direct min-cost flow and existing scheduling methods, and demonstrates substantial computational improvements while scaling to large scheduling instances.
\end{abstract}

\section{Introduction}
The explosion of data-heavy applications, like artificial intelligence and machine learning, has caused a massive spike in the world's need for computing power. Because of this, computing systems consume a massive amount of electricity, which has rapidly driven up our global carbon footprint. Historically, computational resource management has optimized purely for performance metrics or baseline energy efficiency.  A lot of work has been done on online speed scaling to decide at what speed to run a processor as the speed controls the energy used. However, the issues are a bit more complex than simply controlling the speed of the processor \cite{YDS}.

 Today, our power grids use a constantly changing mix of fossil fuels and clean renewables like solar and wind. Because of this, the ``carbon intensity'' (the amount of carbon emitted per kilowatt-hour of electricity generated) fluctuates drastically based on time and geographic location. To make computing truly green, we must transition to carbon-aware scheduling - intelligently routing and deferring computing workloads to align with periods and locations of clean, renewable energy availability. In addition, edge servers are cheaper to run, so decisions have to be made about what can be computed at the edge versus what should be moved to more expensive centralized servers while still  meeting service level agreements and customer needs.

We consider this challenge using a novel setup that involves both edge servers and cloud servers. Edge servers are located physically close to the user, meaning they process data very quickly but have limited capacity. Cloud servers are massive, distant data centers with huge capacity, but sending data to them takes extra time and energy. Within this framework, two fundamental classes of high-level carbon-aware scheduling problems emerge:

\begin{enumerate}
    \item Where to run: When a new job arrives, we need to choose if it should be processed locally on the edge server or moved to the cloud. %Due to access to multiple data centers, jobs may also be moved between data centers. This requires comparing the local grid's carbon cost against multiple cloud's carbon costs, while also factoring in the extra energy/cost/delays needed to send the data over the network.
\item When to run: The core challenge is to deliberately defer or shift flexible workloads to time windows when renewable energy is abundant. This requires navigating time-varying carbon costs to minimize overall emissions without severely degrading performance objectives.
\end{enumerate}

{\bf Why is the problem worth solving?} Cloud operators already act on carbon signals. Google's production system shifts flexible workloads toward cleaner hours using day-ahead forecasts of carbon intensity and demand \cite{radovanovic2023carbonaware}. Trace-driven studies have measured how much there is to gain. Wiesner and coauthors found reductions above ten percent for delay-tolerant jobs in several grids when jobs are given windows rather than fixed start times \cite{wiesner2021letswaitawhile}. Sukprasert and coauthors, analysing three years of hourly data across 123 regions, found that practical savings fall short of idealised upper bounds, that simple policies often capture most of the available reductions, and that the benefits diminish as grids become cleaner \cite{sukprasert2024limitations}. These findings motivate understanding when workload shifting is useful and how close practical schedules come to the optimum.

{\bf Informal Problem Formulation:} We first formally discuss the problems addressed in this paper in Subsection~\ref{sec-defn} and then compare that to prior work in this emerging space in Subsection~\ref{sec-related}. 
Our framework is inspired by 
Cendrier et al \cite{CendrierWBRVC25,cendrier:hal}, and they already showed experimentally why carbon aware scheduling is important and can have significant benefits. They also develop a {\em provably optimal strategy} for a simple case - jobs have release times and deadlines, can be paused and resumed, and form a single chain of precedence constraints. In their followup work\cite{cendrier:hal}, they also consider pre-emptive scheduling with time varying carbon costs {\em but with all jobs having identical release times.}

{\bf Problem Motivation:}
Our interest in the problem arises from basically asking the question - can this framework be generalized in a few ways? Can we allow any precedence constraints and release times? The other is to allow for situations when the set of jobs is simply infeasible on the edge server and there is no option but to move some jobs to the cloud. How do we select which jobs to move? Moving a job to the cloud incurs a job specific fixed cost that has to be factored in, what increases the complexity is that this cost is not fixed but in fact depends on the exact schedule we create. Moving a job that might have been scheduled at a high carbon cost region on the edge server is different from the cost when it is scheduled at a low cost region.

\subsection{Problem Definition and Notations}
\label{sec-defn}

We are given a set $J=\{1,\ldots,n\}$ of jobs, a directed acyclic graph $G=(J,E)$ specifying their precedence constraints, and a time horizon \([0,H)\). Each job $j$ has a processing requirement $p_j>0$, a release time $r_j\in[0,H)$, and a deadline $d_j\in[0,H)$. We will use $i\prec j$ if there is a directed path from $i$ to $j$, and $i\preceq j$ if $i=j$ or $i\prec j$.

There is one unit-speed edge server (local machine), which processes at most one job at a time. Preemption is allowed for free.\footnote{The closely related recent work \cite{CendrierWBRVC25} assumes a simpler model that allows a job to be paused and resumed freely, but this is a weaker assumption than a pre-emptive model where a job can be paused and a different job can be resumed. We believe the simpler model only works because in their case they handle jobs that are in a chain, so in any case of a job is paused, its only that job can be resumed and nothing else is available to be run until this job is completed.} A job must receive all $p_j$ units of processing between its release time and deadline. Let $S_j$ and $C_j$ denote the times at which \(j\) is started and finished, respectively. Precedence constraints require that $C_i\leq S_j$ for all $(i,j)\in E$.

Processing on the edge server at time $t$ incurs a nonnegative cost $c(t)$ per unit of processing time which we refer to as {\em carbon cost}. The interval \([0,H)\) is divided into intervals \([\tau_0=0,\tau_1),[\tau_1,\tau_2),\ldots,[t
\tau_{T-1},\tau_T=H)\), and \(c(t)\) is constant on the interval \([\tau_t,\tau_{t+1})\). Each job induces the same carbon cost, and idle time incurs no cost. For a schedule $\mathcal S$, let $B_j(\mathcal S)$ be the set of times at which $j$ is processed locally on the edge server, and let $B(\mathcal S):=\bigcup_{j\in J}B_j(\cS)$.

We consider two variants of this model.

\begin{enumerate}
\item  In the \emph{edge-only} problem, every job must be processed on the local edge server (if a feasibile schedule exists), and the objective is
\[
\cost_{\mathrm{edg}}(\mathcal S) = \sum_{t=0}^{T-1}\lrm{B(\cS)\cap[\tau_t,\tau_{t+1})}\cdot c(t).
\]

\item In the \emph{cloud-outsourcing problem}, each job $j$ additionally has a nonnegative outsourcing cost $\kappa_j$. We choose a set $A\subseteq J$ of jobs to process entirely on the edge server; every job outside $A$ is executed in the cloud. A cloud job uses no local capacity and executes instantaneously at a chosen time $\theta_j\in[r_j,d_j]$, at cost $\kappa_j$. For such a job, set $S_j=C_j=\theta_j$, so the same precedence inequalities apply to edge and cloud jobs. The cloud has unlimited capacity, so multiple jobs may execute at the same time. If \(\cS_A\) is the schedule of jobs processed locally on the edge server, the objective is
\[
\cost_{\mathrm{cld}}(A,\mathcal S) = \cost_{\mathrm{edg}}(\cS_A)+\sum_{j\in J\setminus A}\kappa_j.
\]
In Fig. \ref{fig:ex}, we give an example of an instance of this problem and a feasible, low-cost solution.
\end{enumerate}

Note that the outsourcing is all-or-nothing, meaning that we cannot move a job to the cloud in the middle of its execution on the edge server. Every job must satisfy its release time, deadline, and precedence constraints regardless of its execution
location. For any given instance, let $\OPT_{\mathrm{edg}}$ and
$\OPT_{\mathrm{cld}}$, respectively, denote the optimal costs for {\em edge-only} and the {\em cloud-outsourcing} problems.

\begin{figure}
\centering
\begin{tikzpicture}

% \node(A) at (1,6) {A};
% \node(B) at (1,5) {B};
% \node(F) at (1,4.5) {F};
% \node(C) at (2,5.5) {C};
% \node(D) at (2,4.5) {D};
% \node(E) at (3,5) {E};

\node(A) at (-3,2.5) {A};
\node(B) at (-3,1.5) {B};
\node(F) at (-3,1) {F};
\node(C) at (-2,2) {C};
\node(D) at (-2,1) {D};
\node(E) at (-1,1.5) {E};

\draw[black, ->] (A) -- (C);
\draw[black, ->] (B) -- (C);
\draw[black, ->] (F) -- (D);
\draw[black, ->] (C) -- (E);
\draw[black, ->] (D) -- (E);

\draw[gray, thick] (0,3) -- (0,0);
\draw[gray, thick] (-0.5,0) -- (6,0);

\fill[Brown] (0,0.25) rectangle (6,0.5);
\fill[Green] (0,0.25) rectangle (3,0.5);
\fill[Green] (4.5,0.25) rectangle (5,0.5);

\draw[black, thick] (0,1.5) -- (2,1.5);
\draw[black, thick] (0,1.4) -- (0,1.6);
\draw[black, thick] (2,1.4) -- (2,1.6);
\node[right] at (2,1.5) {A};
\fill[purple] (0,1.45) rectangle (1,1.55);

\draw[black, thick] (0.5,2.5) -- (3,2.5);
\draw[black, thick] (0.5,2.4) -- (0.5,2.6);
\draw[black, thick] (3,2.4) -- (3,2.6);
\node[left] at (0.5,2.5) {B};
\fill[purple] (1,2.45) rectangle (2,2.55);

\draw[black, thick] (1.5,2) -- (5,2);
\draw[black, thick] (1.5,1.9) -- (1.5,2.1);
\draw[black, thick] (5,1.9) -- (5,2.1);
\node[right] at (5,2) {C};
\fill[purple] (2,1.95) rectangle (3,2.05);

\draw[black, thick] (3.5,1.5) -- (5.75,1.5);
\draw[black, thick] (3.5,1.4) -- (3.5,1.6);
\draw[black, thick] (5.75,1.4) -- (5.75,1.6);
\node[right] at (5.75,1.5) {D};
\fill[purple] (3.5,1.45) rectangle (4.25,1.55);

\draw[black, thick] (3,1) -- (5,1);
\draw[black, thick] (3,0.9) -- (3,1.1);
\draw[black, thick] (5,0.9) -- (5,1.1);
\node[left] at (3,1) {E};
\fill[purple] (4.5,0.95) rectangle (5,1.05);

\draw[black, thick] (3.25,2.5) -- (5.5,2.5);
\draw[black, thick] (3.25,2.4) -- (3.25,2.6);
\draw[black, thick] (5.5,2.4) -- (5.5,2.6);
\node[right] at (5.5,2.5) {F};
\fill[purple!50] (3.25,2.45) rectangle (4.25,2.55);

\node[draw,double,rounded corners](n) at (3.25,3.25) {execute F on the cloud};
\draw[black, ->] (3.25,2.7) -- (n);
\end{tikzpicture}
\caption{A set of jobs with a precedence DAG is presented. The execution windows and job lengths are shown in the timeline. The \textcolor{Green}{low-cost} and \textcolor{Brown}{high-cost} regions are also shown. In order to give a feasible schedule, job F is moved to the cloud. Job E is delayed in order to use a low-cost energy region.}
\label{fig:ex}
\end{figure}
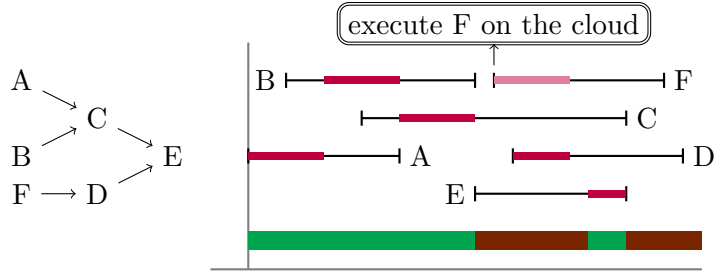

We now describe how we break down the interval \([0,H)\) in the execution of our algorithms.

\begin{defn}[Effective windows]\label{def:effwin}
For each job define
\[
R_j = \max_{i\preceq j}r_i \text{ and } D_j = \min_{j\preceq k}r_k.
\]

These quantities are computed by a forward and a reverse topological pass through $G$. For every $i\prec j$, they satisfy \(R_i\leq R_j\) and \(D_i\leq D_j\). Every feasible schedule in either variant satisfies $S_j\geq R_j$ and $C_j\leq D_j$: a job cannot start before the release of an ancestor or finish after the deadline of a descendant. Local processing on edge server therefore lies in $[R_j,D_j)$, whereas a cloud job may execute at either endpoint of $[R_j,D_j]$. These windows are computed on the original DAG, independently of which jobs are outsourced. 
\end{defn}

\begin{defn}[Atomic intervals]\label{def:atom}
Let \(\sigma=\lrs{r_1,\ldots,r_n}\cup\lrs{d_1,\ldots,d_n}\cup\lrs{\tau_0,\ldots,\tau_T}\). Order \(\sigma\) chronologically so that \(\sigma_0=0\) and \(\sigma_q=H\). The intervals \(I_k=[\sigma_{k-1},\sigma_k)\) partition \([0,H)\) for \(k=1,\ldots,q\), and the carbon cost over each \(I_k\subseteq[\tau_t,\tau_{t+1})\) is \(c(t)\) for each \(t=0,\ldots,T-1\). Each $R_j$ is an original release time and each $D_j$ is an original deadline, so each atomic interval is either wholly inside a job's effective window or disjoint from its interior.
\end{defn}

\subsection{Related Work}
\label{sec-related}
Recent literature has begun to tackle pieces of this puzzle. For example, Penn and Raviv~\cite{PennRaviv2021} study single-machine scheduling under time-of-use electricity tariffs. For preemptive jobs with release times and deadlines, they give
a polynomial-time exact algorithm by reducing the problem to the Hitchcock transportation problem, a special case of
minimum-cost flow. Their formulation assumes independent jobs and does not include precedence constraints.

Lechowicz et al. \cite{lechowicz25} study how to smartly schedule large data processing workloads to align with periods of low-carbon energy. 
%or Cendrier et al. study how to explicitly balance energy consumption with performance when scheduling jobs.
The most relevant work is that of Cendrier et al \cite{CendrierWBRVC25}. In their paper they define the formal model and deal with the special case of a chain of jobs with release times and deadlines with two carbon levels (red/green). This special structure allows for only one job to be active at any point of time, and jobs may be paused and resumed. Their algorithm computes an optimal scheme for maximizing the amount of green intervals it can utilize. Their develop a clever two pass algorithm to find a feasible schedule of minimum carbon cost (if it exists) optimally utilizing green intervals. In the follow-up work \cite{cendrier:hal}, they give a preemptive algorithm for scheduling with variable carbon costs and all jobs released at time 0. This work also gives a linear programming-based algorithm for scheduling on multiple edge servers with free preemption and migration.

Recent work on green computing has mostly studied heuristic-based approaches and only some restricted models of scheduling. Conversely, recent theoretical works like Erlebach et al. \cite{Erlebach} provide mathematical scheduling frameworks for edge-to-cloud task offloading, but they optimize purely for performance-based metrics (like makespan) and completely ignore carbon costs. In the stochastic setting, recent work has focused on improved heuristics scheduling with time-varying holding costs.

\subsection{Key Results}
% \textcolor{red}{@emmi- add a para here for relevance of pause-and-resume hardness citing the exisiting paper on a chain}. 
As mentioned, the previous work by Cendrier et al. in \cite{CendrierWBRVC25} uses the limitation that a started job must be finished before the next starts. As they prove, while this model does have an exact polynomial-time algorithm when the job order is part of the input, it does not admit a constant-factor approximation without an order. Therefore, in this paper, we focus on the {\em preemptive} model. Note that preemption is available for free (without pre-emption, the problems in this space are generally NP-hard).

%Key results. {\em Write these out formally}

%\begin{enumerate}
%    \item Generalization of the window scheduling model.
%    \item Dealing with precedence constraints.
%    \item Relating moving to cloud with Tardiness.
%    \item Experimental evaluation and comparison.
%\end{enumerate}

We first study the {\em edge-only} problem to establish polynomial time solvability.
\begin{thm}\label{thm:exact}
The edge-only problem admits a polynomial-time algorithm that returns a minimum-cost feasible schedule or correctly reports infeasibility.
\end{thm}

The first important observation is that the precedence constraints can be restored without changing when the server is active. In particular, after replacing every job's original release time and deadline with the effective release time-deadline window, any feasible schedule that ignores the precedence constraints can be rearranged to satisfy the original DAG. This extends the flow-based algorithm of Penn and Raviv \cite{PennRaviv2021} for independent jobs to arbitrary precedence DAGs.

We also improve on the work of Penn and Raviv~\cite{PennRaviv2021} by obtaining a faster exact algorithm. Rather than solving a Hitchcock transportation problem, our algorithm enables interval capacity in nondecreasing order of carbon cost and incrementally augments a maximum flow. These updates increase only capacities on arcs leaving the source, while all other capacities remain fixed. The resulting
sequence therefore fits the monotone parametric maximum-flow framework
of Gallo et al. \cite{GGD}, which reuses the preflow
and distance labels across phases. This gives an $O(n^3)$ running time, improving on the $O(n^4\log n)$ bound stated by Penn and Raviv \cite{PennRaviv2021}. 

%{\em Can we improve this paragraph?}
%The key structural result is that, after tightening job windows using the precedence DAG, an independent-job schedule can be rearranged to satisfy every precedence constraint without changing the server’s busy-time set. Since carbon cost depends only on this set, precedence restoration does not increase the objective. We then formulate the allocation of processing to time intervals as a flow problem. Rather than solving a general minimum-cost-flow instance, our algorithm enables interval capacity in nondecreasing order of carbon cost and incrementally augments a maximum flow. Retaining the residual network allows later phases to rearrange earlier assignments without discarding the computation already performed.
%Note that, our result improves on the work of Penn and Raviv~\cite{PennRaviv2021} in two respects. We obtain a faster
%exact algorithm for the independent-job setting they study, and we extend the model to accommodate arbitrary precedence constraints while retaining polynomial-time exact solvability.

We consider the problem with {\em cloud outsourcing}. Allowing outsourcing makes the problem NP-hard, even for independent jobs with a common release time and deadline and zero local carbon cost. We show this hardness via a reduction from 0-1 knapsack problem. Thus we focus on an approximation algorithm. We obtain a fully polynomial-time approximation scheme: for every \(\varepsilon\in(0,1)\), the algorithm computes a solution of cost at most \((1+\varepsilon)\mathrm{OPT}_{\mathrm{cld}}\) in time polynomial in the binary input length and \(1/\varepsilon\) (See Theorem \ref{thm:cloud}).

\begin{thm}
\label{thm:cloud}
The cloud-outsourcing problem admits a fully polynomial-time approximation scheme: for every $\varepsilon\in(0,1)$, the algorithm returns a mixed schedule $(A,\mathcal S)$ satisfying
\[
  \cost_{\mathrm{cld}}(A,\mathcal S)
  \leq (1+\varepsilon)\OPT_{\mathrm{cld}}.
\]
Its run time is polynomial in the binary input length and $1/\varepsilon$.
\end{thm}

We give an approximation-preserving reduction to single-machine preemptive scheduling with rejection penalties, equivalently minimizing the weighted number of tardy jobs\footnote{ In this problem, jobs have to be scheduled pre-emptively and have processing times, release times and deadlines, and the goal is to minimize the weight of the jobs that miss their deadlines.}. 
However we cannot use a direct reduction as the cost to schedule jobs is time dependent in our setting.
The reduction assigns rejection penalties to original jobs to represent outsourcing and introduces a compact collection of dummy jobs to represent local time capacity. Rejecting dummy jobs makes capacity available to original jobs and charges the corresponding carbon cost. Applying the FPTAS of Pruhs and Woeginger \cite{PruhsW07}, based on Lawler’s dynamic program \cite{Lawler}, gives the stated guarantee.

We also show that the single-server assumption matters: with two edge servers and no migration between them, even the edge-only problem admits no approximation guarantee polynomial in the number of jobs when carbon-cost variation is unrestricted, unless \(P=NP\). This holds despite allowing preemption on each server and giving both servers the same carbon-cost function.

Finally, our experiments show that the max-flow formulation also provides a substantial computational advantage in practice. On synthetic workloads, \textsc{MaxFlow-Passes} is $4.96\times$ faster than direct minimum-cost flow on average. Reusing the residual network across carbon-cost phases provides an additional $2.07\times$ improvement over restarting the max-flow computation at each phase. The implementation scales to flow networks with tens of millions of job--interval edges. Experiments on workloads derived from the Alibaba cluster trace show a similar advantage, with \textsc{MaxFlow-Passes} reaching approximately $10\times$ the speed of \textsc{MinCostFlow} on some of the largest tested instances. We also evaluate Lawler's exact dynamic program \cite{Lawler} for the cloud-outsourcing reduction and observe the expected pseudopolynomial growth with the total penalty weight.

% \textcolor{red}{@cecelia - add experimental results here.}

%\subsection{Outline of Paper}

\subsection{Organization of the paper} 
%Section \ref{sec:cloud} establishes the hardness for the pause-and-resume model. 
Section \ref{sec:exact} presents the exact edge-only algorithm. Section \ref{sec:cloud} presents the approximation scheme for cloud outsourcing. Section \ref{multiedge} proves the hardness result for multiple edge servers without migration, and Section \ref{sec:exp} presents our experimental evaluation: we compare the exact solvers, study the benefit of residual reuse and scalability, evaluate the algorithms on Alibaba cluster traces, and examine the scalability of Lawler's algorithm for the cloud-outsourcing reduction.

\section{Exact scheduling on a single edge server}\label{sec:exact}

% what if we have a constant number of distinct cost values

We start by giving an exact algorithm for preemptively scheduling jobs on a single edge server with a piecewise carbon cost function. Let \(I_1,\ldots,I_q\) be the atomic intervals defined by all release times, deadlines, and changes in carbon cost. Recall that carbon cost \(c_k\) is uniform over the interval \(I_k\). We construct a flow network \(\cF\) to model our scheduling problem. Our network has a source node \(s\) and a sink node \(t\). For each interval \(I_k\), we have a node \(v_k\), and for each job \(j\), we have a node \(u_j\). We then define the arcs:
\begin{center}
\begin{tabular}{lc|c|c}
arc & & capacity & cost \\
\hline
\(s\to v_k\) & & \(|I_k|\) & \(c_k\) \\
\(v_k\to u_j\) & if \(I_k\subseteq[R_j,D_j)\) & \(p_j\) & 0 \\
\(u_j\to t\) & & \(p_j\) & 0
\end{tabular}
\end{center}

Let \(P=\sum_{j\in J}p_j\) be the total processing required. If \(\cF\) has a flow of value \(P\), then we can construct a feasible schedule. This saturating flow assigns \(f(v_k,u_j)\) units of job \(j\) to be completed during interval \(I_k\). In order to convert the flow into a schedule, we arbitrarily order the jobs \(j\) with \(f(v_k,u_j)>0\) and complete \(f(v_k,u_j)\) units of each job in that order. When we do this for each interval, we get a preemptive schedule that completes \(p_j\) units of each job during the interval \([R_j,D_j)\). This gives us a preemptive schedule with cost \(\sum_{k=1}^qc_kf(s,v_k)\). Conversely, given a feasible preemptive schedule, we can return a flow network with value \(P\) by allowing \(f(s,v_k)\) to be the amount of time the edge server is running during interval \(I_k\) and \(f(v_k,u_j)\) to be the amount of job \(j\) completed during interval \(I_k\); set \(f(u_j,t)=p_j\). This returns a flow with cost \(\sum_{k=1}^qc_kf(s,v_k)\). Therefore, we have an equivalence between the costs of minimum-cost schedules and minimum-cost flows.

\subsection{Post-processing for precedence constraints}

An important observation is that this algorithm considers feasibility only in terms of the effective windows, not the precedence constraints.
\begin{lem}\label{lem:swap}
Any schedule that is feasible with respect to the effective windows can be converted into one that follows the precedence constraints and does not increase carbon cost.
\end{lem}

\begin{proof}
Let \(\cS\) be a feasible schedule returned by the flow algorithm. Take jobs \(i,j\) with \(i\prec j\). In \(\cS\), \(i\) and \(j\) are completed during the intervals \([S_i,C_i)\) and \([S_j,C_j)\). We want to make an adjustment \(g_{i,j}(\cS)=\cS'\) such that \(C_i'\leq S_j'\), meaning that \(i\) completes before \(j\). Let \(B_{i,j}(\cS)\) be the collection of disjoint intervals during which the edge server is running jobs \(i\) and \(j\). The adjusted schedule \(g_{i,j}(\cS)\) completes job \(i\) during the first \(p_i\) units of \(B_{i,j}(\cS)\) and then completes job \(j\) during the last \(p_j\) units of \(B_{i,j}(\cS)\). This guarantees that \(C_i'\leq S_j'\). Moreover, since \(R_i\leq R_j\) and \(D_i\leq D_j\), and job \(i\) is completed earlier while job \(j\) starts later, we maintain that \([S_i',C_i')\subseteq[R_i,D_i)\) and \([S_j',C_j')\subseteq[R_j,D_j)\).

We complete this process in topological order. For each job \(i\), we make the swap \(g_{i,j}(\cS)\) for each \(j\) such that \(i\preceq j\). Therefore, we complete a polynomial number of swaps. After we have processed job \(i\), we have \(C_i'\leq S_j'\) for each \(j\succeq i\). We need to show that this invariant is maintained for this pair in later steps of the algorithm. Assume \(h\prec i\prec j\), and we are processing \(i\). Before completing this step, we have \(C_h\leq S_i\) and \(C_h\leq S_j\). After processing, \(S_i'=\min\{S_i,S_j\}\geq C_h\), so there is no retroactive violation of precedence constraints with \(h\).

Observe that the initial schedule \(\cS\) and the final schedule \(\cS'\) occupy the edge server during the exact same intervals of time. As a result, these schedules have identical costs.
\end{proof}

\subsection{Implementation without minimum-cost flow}
Due to the simplified nature of the cost function defined for \(\cF\), we can avoid using a min-cost flow algorithm by gradually allowing flow through vertices \(v_k\) in order of increasing cost \(c_k\). Start with the capacities of arcs originating from \(s\) set to 0. At each step \(h=1,\ldots,q\), allow the capacity of arc \(s\to v_h\) to change to \(|I_h|\) in order of increasing cost, breaking ties arbitrarily. Push flow through the network until no augmenting paths remain. The residual flow is retained after each step of the algorithm. Let \(F_h\) be the flow value at the end of step \(h\), meaning that we can schedule \(F_h\) units of work on the edge server during the cheapest \(h\) intervals. We have \(F_0=0\). If \(F_q<P\), then there is no feasible schedule.

\begin{lem}
Assume that \(F_q=P\). Let \(y_k\) be the amount of flow on \(s\to v_k\) after the incremental max flow procedure, and let \(z_k\) be the amount of flow on \(s\to v_k\) induced by any other min-cost flow of value \(P\). Then, \(\sum_{h=1}^qc_hz_h>\sum_{h=1}^qc_hy_h\).
\end{lem}

\begin{proof}
We first show an equivalence between the intermediate flow values \(F_h\) and the final flow \(y\). During step \(h\), the algorithm looks for an augmenting path \(s\to t\). However, it cannot use an arc \(v_k\to s\), as doing so would simply return the flow to \(s\). Therefore, after the completion of this step, the flow through the first \(h\) arcs \(s\to v_k\) cannot decrease, and \(\sum_{k=1}^hy_k>F_h\). 

For the reverse inequality, we consider the final flow and restrict it to the first \(h\) intervals. This is a feasible flow for scheduling on the cheapest \(h\) intervals. By the definition of \(F_h\), we have \(\sum_{k=1}^hy_k<F_h\). Therefore, we conclude that, for each \(h=1,\ldots,q\), we have 
\begin{equation}\label{eq:yF}
\sum_{k=1}^hy_k=F_h.
\end{equation}

Now, we consider our other flow \(z\). Let \(z_0=0\). We have 
\begin{align*}
&\phantom{=}\quad\sum_{h=1}^qc_hz_h = \sum_{h=1}^q c_h\lrp{\sum_{k=0}^hz_k-\sum_{k=0}^{h-1}z_k} \\
&= -c_1\sum_{k=0}^0z_k +\sum_{h=1}^{q-1}\lrp{\lrp{c_h-c_{h+1}}\sum_{k=0}^hz_k} +c_q\sum_{k=0}^qz_k \\
&> \sum_{h=1}^{q-1}\lrp{c_h-c_{h-1}}F_h +c_qF_q \tag{Def. of \(F_h\)} \\
&= \sum_{h=1}^qc_h\lrp{F_h-F_{h-1}} = \sum_{h=1}^qc_hy_h \tag{Eq. \eqref{eq:yF}},
\end{align*}
and our incremental max flow algorithm finds a flow of value \(P\) with cost no more than any min-cost flow with value \(P\).
\end{proof}

By using the incremental flow algorithm and then correcting the schedule using Lemma \ref{lem:swap}, we prove Theorem \ref{thm:exact}.
\section{Approximation with cloud outsourcing}
\label{sec:cloud}

The cloud-outsourcing problem is NP-hard by a straightforward reduction from 0–1 Knapsack, even for independent jobs with a common release time, a common deadline, and zero carbon cost: selecting jobs to execute locally corresponds to packing items, with processing times as sizes and avoided outsourcing costs as values. Hardness persists even when \(\kappa_j=\lambda p_j\) for any fixed \(\lambda>0\), by a reduction from Subset Sum. We next obtain an FPTAS through a reduction to single-machine preemptive scheduling with rejection penalties.

\subsection{Reduction to the penalty variant}

%We first show that precedence can be restored without changing busy times even with cloud outsourcing and, then show how the problem can be reduced to the problem of minimizing the total penalty of rejected jobs in preemptive single-machine scheduling.

We first show that we have a version of Lemma \ref{lem:swap} that extends to the cloud outsourcing variant. Fix any job set $A$ run on the edge server with a feasible independent job schedule in the effective windows. Apply the order-restoring construction from Lemma~\ref{lem:swap} to the DAG \(G\) restricted to the jobs in \(A\). This restores every edge server ancestor-descendant relation without changing the intervals of time at which the edge server is running. Now, for every cloud job $j\notin A$, set
\begin{equation}\label{eq:cloud-execution}
\theta_j = \max\lrp{\{R_j\}\cup\{C_i:i\in A,\ i\prec j\}}.
\end{equation}
Every edge server ancestor $i$ satisfies $C_i\leq D_i\leq D_j$, so $r_j\leq R_j\leq\theta_j\leq D_j\leq d_j$. Edge-to-cloud precedences follow directly from \eqref{eq:cloud-execution}. Cloud-to-cloud precedences hold because both effective releases and the sets of edge server ancestors are nested along precedence paths. Finally, we maintain cloud-to-edge precedences. If a cloud job $i$ precedes a local job $j$, then $R_i\leq R_j\leq S_j$, and every local ancestor of $i$ is also a local ancestor of $j$ and finishes by $S_j$. Thus $\theta_i\leq S_j$. This proves that any feasible independent schedule of $A$ can be extended to a feasible mixed schedule.

\begin{defn}
In the penalty-scheduling problem, jobs have parameters \((p_j,r_j,d_j,w_j)\), where the first three are defined as above, and \(w_j\) is a nonnegative penalty. If a job does not complete within the window \([r_j,d_j)\), we pay the penalty \(w_j\). The objective is to select a subset of jobs to feasibly schedule such that the total penalty paid is minimized.
\end{defn}

Pruhs and Woeginger~\cite{PruhsW07} give an FPTAS
for this minimization problem with arbitrary release times and due dates. This is built on the exact algorithm due to Lawler based on dynamic programming \cite{Lawler}.

Next, we will reduce our {\em cloud-outsourcing} problem to the penalty-scheduling problem. We asumme that all processing requirements, releases, deadlines, and carbon-change times are integers. 

%By a common rescaling, we may assume that all processing requirements, releases, deadlines, and carbon-change times are integers. Specifically, if $Q$ is a common denominator, replace time $t$ by $t'=Qt$ and replace the carbon rate by
%\[
%c'(t')=\frac{c(t'/Q)}{Q}.
%\]
%Then
%\[
%c'(t')\,dt'=c(t)\,dt,
%\]
%so all schedule costs are preserved.

\noindent\textbf{Reduction:} For every original job $j\in J$, create a job in the penalty instance with
parameters $(p_j,R_j,D_j,w_j=\kappa_j)$. Next we will construct some dummy jobs.

For each interval $I_k$ of length $\ell_k$, we want to create $O(\log\ell_k)$ dummy jobs such that the sum of processing times of these jobs is $\ell_k$ and the subset sums of the processing times include every integer in $\{0,\ldots,\ell_k\}$. Let \(m=\lfloor\log_2(\ell_k+1)\rfloor\). A suitable set of job lengths is \(P_k=\lrs{2^0,2^1,\ldots,2^{m-1},\ell_k-2^m+1}\). The subsets of the powers of two sum to all values in $\lrs{0,1,\ldots,2^m-1}$. Since $0\leq b\leq2^m-1$, adding allows us to sum to all values in $\{0,\ldots,\ell_k\}$.

For each job length \(h\in P_k\) create a dummy job $g_{k,h}$ with
\[
p(g_{k,h})=h, r(g_{k,h})=t_{k-1}, d(g_{k,h})=t_k, w(g_{k,h})=c_kh.
\]

The dummy jobs for $I_k$ have total processing requirement $\ell_k$. Scheduling some of them after their deadline releases capacity in that interval and incurs a penalty equal to the carbon cost of the released capacity. Let $\OPT_{\mathrm{pen}}$ be the minimum total penalty in the constructed instance.

\subsection{Proof of correctness}

\begin{lem}\label{lem:eq-opt}
\(\OPT_{\mathrm{pen}}=\OPT_{\mathrm{cld}}\).
\end{lem}

\begin{proof}
We start by showing that $\OPT_{\mathrm{pen}}\leq\OPT_{\mathrm{cld}}$. Take an optimal mixed schedule and let $A$ be the job set scheduled on edge server. For this fixed set, solve the minimum-cost interval-job flow problem on the effective windows. Its cost is at most the original edge server carbon cost. All capacities and job requirements are integral, so it has an integral optimal flow. Let $y_k$ be its total processing volume in $I_k$; then $y_k\in\{0,\ldots,\ell_k\}$.
    
Accept the original jobs in $A$ and reject the others. In interval $I_k$, reject a subset of dummy jobs of total length exactly $y_k$, which exists by the subset-sum property. Accept the remaining dummy jobs, whose total length is $\ell_k-y_k$. The original job allocation and accepted dummy jobs now have total volume $\ell_k$ in $I_k$. Every participating job is eligible throughout that interval, so these allocations can be serialized to give a feasible accepted-job schedule. Its penalty is
\[
  \sum_{j\notin A}\kappa_j+\sum_{k=1}^q c_ky_k
  \leq\OPT_{\mathrm{cld}}.
\]

Now, we show that $\OPT_{\mathrm{pen}}\geq\OPT_{\mathrm{cld}}$. Consider any feasible accepted job schedule of total rejection penalty
$W$. Let $A$ be its accepted original jobs. In interval $I_k$, let $z_k$ be the total processing requirement of the rejected dummy jobs, and let $y_k$ be the processing volume of jobs in $A$. All accepted dummy processing for this interval must occur within $I_k$, so capacity gives
\[
  y_k+(\ell_k-z_k)\leq\ell_k,
  \qquad\text{hence}\qquad y_k\leq z_k.
\]
Delete the dummy jobs, keep $A$ on edge server, and outsource $J\setminus A$ to cloud. Restore precedence and insert cloud jobs using~\eqref{eq:cloud-execution}. The edge server busy-time set is preserved, and nonnegativity of the tariff implies
\[
  \cost_{\mathrm{cld}}(A,\mathcal S)
  =\sum_{j\notin A}\kappa_j+\sum_{k=1}^q c_ky_k
  \leq\sum_{j\notin A}\kappa_j+\sum_{k=1}^q c_kz_k
  =W.
\]
The cost of any feasible solution can only be greater than the cost of optimal schedule, therefore, $\cost_{\mathrm{cld}}(A,\mathcal S)=W\geq\OPT_{\mathrm{pen}}$.
\end{proof}

In fact, the proof shows that the backward direction preserves every multiplicative approximation guarantee. 

By applying the reduction above and then reordering to follow precedence constraints, we achieve Theorem \ref{thm:cloud}.

\section{Hardness for Multi-edge-server scheduling}
\label{multiedge}

In this section, we show that the multi-edge-server scheduling problem is hard to approximate even with only two edge servers
and edge-only setting. Each server has unit speed and processes at most one job at a time. Preemption is allowed for free, but migration between servers is forbidden: every job must receive all of its processing on a single server. The objective is to minimize the total carbon cost across the servers while satisfying all release times, deadlines, and precedence constraints.

Unlike the single-server edge-only problem, which admits an exact
polynomial-time algorithm, this problem admits no approximation
guarantee polynomial in the number of jobs when carbon-cost
variation is unrestricted, unless $\mathrm{P}=\mathrm{NP}$.

\begin{thm}
\label{thm:multiple-edge-hardness}
Unless $\mathrm{P}=\mathrm{NP}$, the edge-only scheduling problem on two servers, with free preemption but no migration, admits no polynomial-time approximation algorithm whose approximation ratio is bounded by a polynomial in the number of jobs $n$, when the ratio between the largest and smallest carbon costs is unrestricted.
\end{thm}

In fact, the theorem holds even for independent jobs with a common release time and deadline, when both servers share the same strictly positive carbon-cost function consisting of two constant-cost intervals, and a feasible schedule is guaranteed to exist.

\begin{proof}

Suppose that a polynomial-time algorithm has approximation factor at most $A n^b$, for fixed integers $A\ge 1$ and $b\ge 0$. Such a bound covers every polynomially bounded approximation factor, including constant and logarithmic factors.

Reduce from \textsc{Partition}. Given positive integers $a_1,\ldots,a_n$ with $\sum_j a_j=2B$, create two servers and jobs
\[
  p_j=a_j,\qquad r_j=0,\qquad d_j=2B.
\]
Define $L=2BA n^b+1$ and give both servers the carbon function
\[
  c_1(t)=c_2(t)=
  \begin{cases}
    1, & 0\le t<B,\\
    L, & B\le t<2B.
  \end{cases}
\]

If the \textsc{Partition} instance is a YES-instance, assign a subset of total size $B$ to the first server and its complement to the second. Both servers finish by time $B$, at total cost $2B$. Since the total
processing requirement is $2B$ and all carbon rates are at least one, the optimum is exactly $2B$.

If the instance is a NO-instance, consider any feasible schedule. Each server's assigned load is an integer sum of whole job sizes. The loads sum to $2B$, and they cannot both equal $B$. One server therefore has load at least $B+1$. It can execute at most $B$ units before
time $B$, so at least one unit must be processed during the expensive interval. If $y$ is the total processing performed after time $B$, then $y\ge 1$ and
\[
  C(S)=(2B-y)+Ly=2B+(L-1)y\ge 2B+L-1>L.
\]
Preemption on a fixed server does not change this load argument.

On a YES-instance the approximation algorithm returns cost at most $2BA n^b=L-1$, whereas on a NO-instance every feasible schedule costs more than $L$. Comparing the returned cost with $L$ decides \textsc{Partition} in polynomial time, a contradiction.
\end{proof}

\section{Experiments}
\label{sec:exp}
% experimental setup 
%     implementations
%     synthetic workload generation
%     metrics
%     hardware
% comparison of exact solvers
%     maxflow-pass vs. mincost flow
%     residual reuse ablation
%     flow distribution across passes
%     rearrangment across passes
% scability
%     scaling with number of jobs
%     scaling with flow-network size
%     effect of interval density
%     CPU vs. GPU implementation
% evaluation on alibaba cluster traces
% cloud outsourcing reduction and lawler's algorithm

We experimentally evaluate the algorithms introduced in previous sections. Our experiments have four goals. First, we compare the running time of the max-flow based algorithm of Section~\ref{sec:exact} with a direct minimum-cost formulation. Second, we study how the running time changes as the size and structure of the scheduling instance and resulting flow network change. Third, we evaluate the algorithms on the production-derived workloads from the Alibaba cluster trace. Finally, we evaluate the reduction of Section~\ref{sec:cloud} and the running time of Lawler's exact dynamic program.

We use synthetic instances for the first two sets of experiments. This allows us to control system utilization, carbon composition, and flow-network size independently. The Alibaba experiments complement these results with workloads derived from production traces. 
%Our source code and experimental artifacts are publicly available~\cite{githubpage}.

\subsection{Experimental Setup}
\label{sec:exp-setup}

\paragraph{Implementations.} 
We compare two exact implementations of the edge-only problem. The first, denoted \textsc{MaxFlow-Passes}, is the algorithm described in Section~\ref{sec:exact}. It orders the carbon intervals by cost, progressively enables their capacity. The current flow and residual network are retained between successive phases. The CPU implementation uses Dinic's algorithm to compute the required max flows.

The second implementation, denoted \textsc{MaxFlow-Restart}, is an ablation of \textsc{MaxFlow-Passes}. It uses the same sequence of carbon levels, but discards the residual graph after each phase. The network is rebuilt before solving the next prefix problem. The cheaper-capacity locks obtained in earlier phases are retained. Thus, the two implementations differ in whether the previous residual computation is reused.

%cheaper-capacity locks obtained in earlier phases are retained
% , so the
% ablation follows the same cheapest-first objective while intentionally
% discarding the computational state that \textsc{MaxFlow-Passes} reuses.
% For every generated instance, we verify that
% \textsc{MaxFlow-Restart} and \textsc{MaxFlow-Passes} return the same
% normalized carbon cost.

The third implementation, denoted \textsc{MinCostFlow}, constructs the
same job--interval network. It assigns cost $c_k$ to the capacity of interval
$I_k$, and directly computes a minimum-cost flow of value
$
    P=\sum_{j\in J}p_j.
$
All three methods compute the same optimal objective. Their objective values agree on every tested instance. The experiments focus on their computational behavior. \footnote{All CPU experiments were executed on a KVM virtual machine with 24 x86-64 virtual CPUs based on an Intel Xeon Cascade Lake processor
and 94\,GB of memory.
Unless otherwise stated, we use random seed $42$. End-to-end algorithm time includes interval splitting, graph construction, flow computation, recovery of the flow assignment, and construction of the final schedule. Instance generation and file I/O are excluded.}

\paragraph{Synthetic jobs.}
For the utilization and green-share experiments, each instance contains between $35$ and $50$ jobs. After choosing the number of jobs $n$, we set the horizon to $H=6n$. Each processing requirement is an integer in ${1,\ldots,6}$.

For a target-utilization range $[\rho_{\min},\rho_{\max}]$, the generator selects a total processing requirement within
$
[\rho_{\min}H,\rho_{\max}H]
$
and distributes these processing units among the jobs. The realized utilization is
$
\rho=\frac{\sum_{j\in J}p_j}{H}.
$
The choice $H=6n$ ensures that the highest utilization range remains attainable under the bound $p_j\leq 6$.

Feasibility is guaranteed by construction. We first create a non-overlapping witness schedule containing all jobs. For job $j$ with the witness piece $[s_j,e_j)$, its release time is sampled from
$[0,s_j]$ and its deadline from $[e_j,H]$.  The witness schedule is not
provided to the algorithms. They receive only the resulting
$(r_j,d_j,p_j)$ values.  Jobs are shuffled before the instance is passed
to the solvers.

The carbon timeline is partitioned into between $5$ and $12$ piecewise-constant intervals unless otherwise specified. We use three normalized carbon costs,
$
c_G=0,\qquad c_B=1,\qquad c_R=2.
$
For the two-dimensional parameter sweep, the green share ranges from $0.20$--$0.25$ to $0.80$--$0.85$ in increments of $0.05$. The non-green portion is randomly divided between brown and red capacity. Target utilization ranges from $0.30$--$0.35$ to $0.95$--$1.00$, also in increments of $0.05$. This gives $14\times13=182$ utilization/green-share configurations. We generate $30$ independent instances for each configuration and repeat each timing measurement $20$ times.

\paragraph{Metrics.}
Our primary computational metric is wall-clock running time.  We report
\[
    \mathrm{Ratio}_{\mathrm{MCF}}
    =
    \frac{T_{\mathrm{MinCostFlow}}}
         {T_{\mathrm{MaxFlow-Passes}}},
\]
which measures the running-time difference between the two exact formulations.  We additionally report
\[
    \mathrm{Ratio}_{\mathrm{Restart}}
    =
    \frac{T_{\mathrm{MaxFlow-Restart}}}
         {T_{\mathrm{MaxFlow-Passes}}},
\]
which isolates the effect of retaining the residual graph between phases.  Values larger than one favor
\textsc{MaxFlow-Passes}.

To characterize where the proposed algorithm performs its work, let
$F_G$, $F_B$, and $F_R$ denote the cumulative flow after the green,
brown, and red phases, respectively.  For every feasible instance,
$F_R=P$.  We report the fraction of total processing introduced in each
phase:
\[
    \eta_G=\frac{F_G}{P},\qquad
    \eta_B=\frac{F_B-F_G}{P},\qquad
    \eta_R=\frac{F_R-F_B}{P}.
\]
Thus $\eta_G+\eta_B+\eta_R=1$. 

Finally, we measure how much of an earlier assignment is rearranged after
a more expensive carbon level is enabled.  Let
$x^{(h)}_{jk}$ denote the amount of job $j$ assigned to atomic interval
$I_k$ after phase $h$.  For the transition from phase $h-1$ to phase $h$,
we define the rerouted volume over the previously enabled intervals as
\[
    R_h =
    \frac{1}{2}
    \sum_{(j,k)\in \mathcal{E}_{h-1}}
    \left|
       x^{(h)}_{jk}-x^{(h-1)}_{jk}
    \right|,
\]
where $\mathcal{E}_{h-1}$ contains job--interval pairs belonging to
carbon levels enabled before phase $h$. We report $R_h$ normalized by the amount
of flow present before the transition.

% For the experiments with scheduling heuristics, we report carbon cost
% relative to the optimum:
% \[
%     \mathrm{Ratio}_{\mathrm{cost}}(A)
%     =
%     \frac{\mathrm{cost}(A)}
%          {\mathrm{OPT}_{\mathrm{edg}}},
% \]
% where $\mathrm{OPT}_{\mathrm{edg}}$ is computed by
% \textsc{MaxFlow-Passes}.

\subsection{Comparison of Exact Solvers}
\label{sec:exp-mcf}

\paragraph{Running time.}
\begin{figure}[t!]
    \centering
    \includegraphics[width=0.75\columnwidth]
        {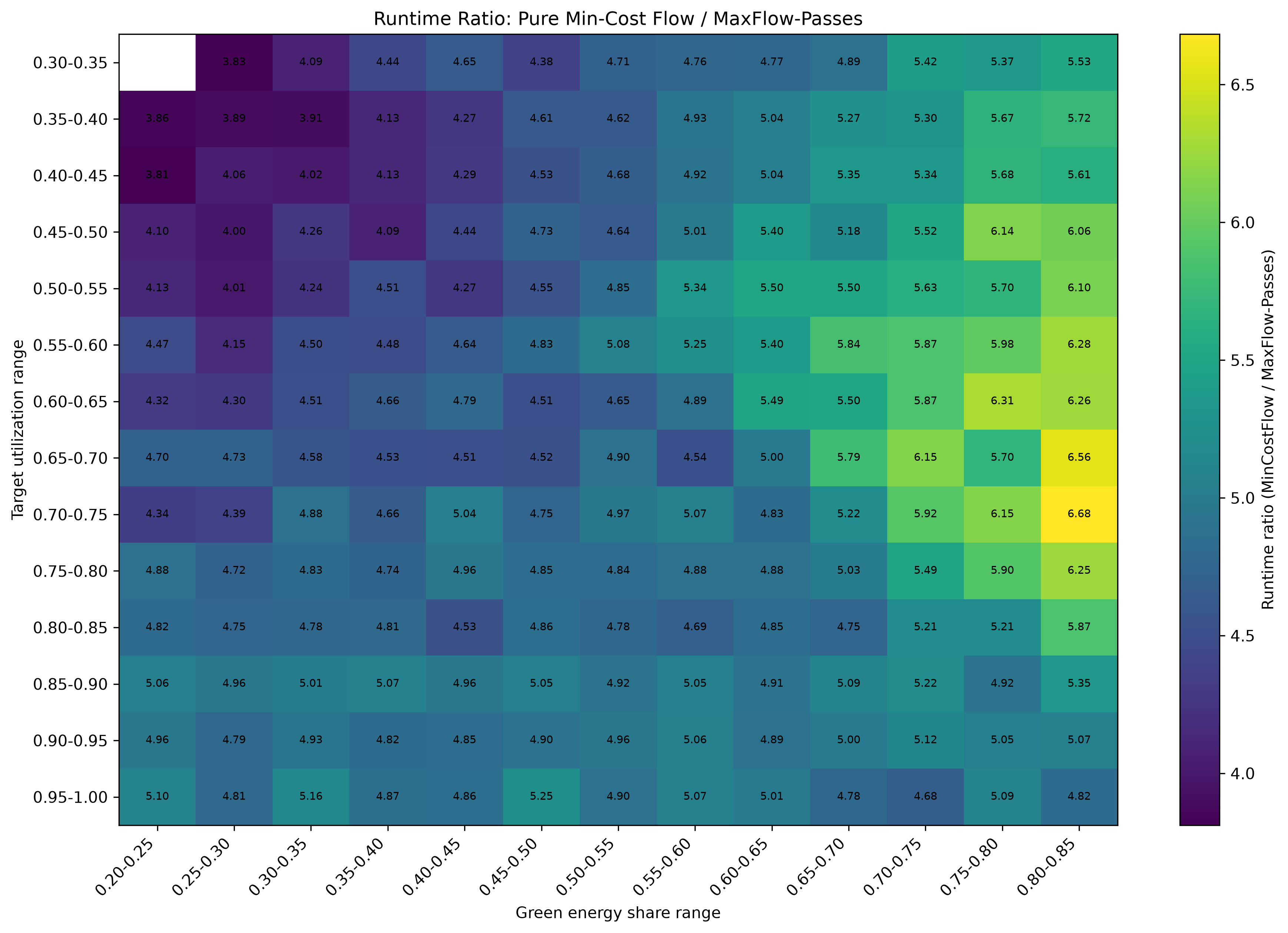}
    \caption{Runtime ratio
    $T_{\mathrm{MinCostFlow}}/T_{\mathrm{MaxFlow-Passes}}$
    as target utilization and green-energy share vary.  Each cell is
    averaged over all instances in the corresponding parameter
    configuration.  Values larger than one favor
    \textsc{MaxFlow-Passes}.}
    \label{fig:runtime-parameter-sweep}
\end{figure}

We first compare \textsc{MaxFlow-Passes} with the direct \textsc{MinCostFlow} formulation. We vary target utilization and green-energy share, since these parameters determine both the total demand placed on the machine and the amount of processing that can be assigned during the earlier, lower-carbon phases.

Figure~\ref{fig:runtime-parameter-sweep} reports
\[
    \frac{T_{\mathrm{MinCostFlow}}}
         {T_{\mathrm{MaxFlow-Passes}}}
\]
for each configuration. Across all instances, \textsc{MinCostFlow} requires $4.961\times$ the runtime of \textsc{MaxFlow-Passes} on average. The median ratio is $4.880\times$. At the configuration level, the average ratio ranges from approximately $3.81\times$ to $6.68\times$.

The difference also varies with the amount of green capacity. At low and moderate utilization, the advantage of \textsc{MaxFlow-Passes} generally increases with the green share. The ratio exceeds $6\times$ in several configurations. At high utilization, the difference is more stable because most available capacity must eventually be used. \textsc{MaxFlow-Passes} remains faster throughout the tested parameter space.

\subsection{Benefit of Residual Reuse}
\label{sec:exp-pass-machanism}

\begin{figure}[t]
    \centering
    \includegraphics[width=0.75\columnwidth]
        {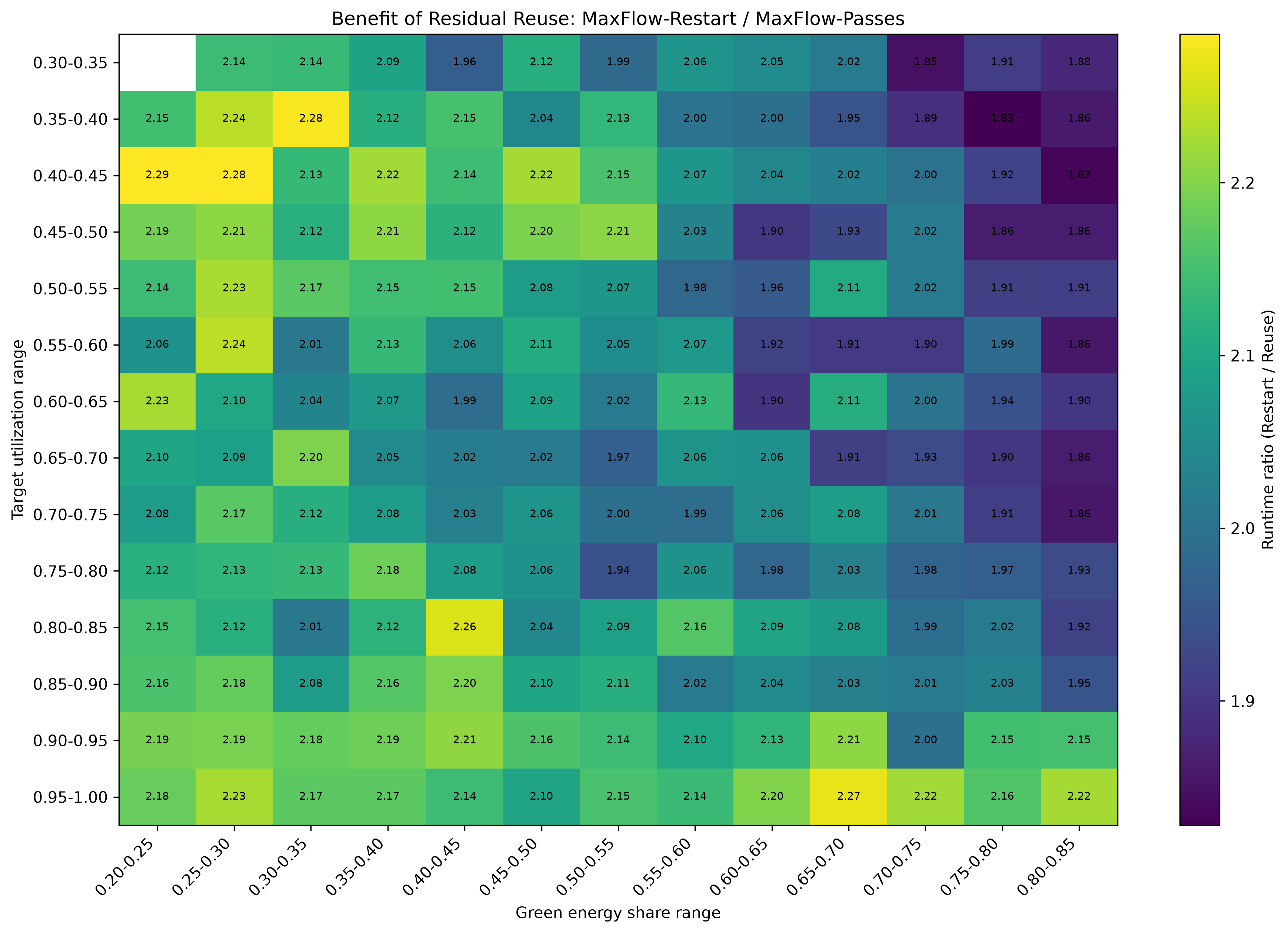}
    \caption{Benefit of retaining the residual network.  Each cell shows
    $T_{\mathrm{MaxFlow-Restart}}/
     T_{\mathrm{MaxFlow-Passes}}$
    averaged over the corresponding utilization/green-share
    configuration.}
    \label{fig:residual-reuse}
\end{figure}

We next isolate the effect of preserving the residual network between carbon-cost phases. \textsc{MaxFlow-Passes} retains the flow after the green phase, enables brown capacity, and continues augmenting the same residual network. The same procedure is repeated when red capacity is enabled.

\textsc{MaxFlow-Restart} follows the same cheapest-first sequence, but rebuilds the network whenever a new carbon level is enabled. Figure~\ref{fig:residual-reuse} reports 
\[\frac{T_{\mathrm{MaxFlow-Restart}}}{T_{\mathrm{MaxFlow-Passes}}}.\]
Across all individual instances, restarting requires $2.069\times$ the
runtime of the incremental implementation on average, with a median
ratio of $2.009\times$.  The average ratio within individual
utilization/green-share configurations ranges from approximately
$1.83\times$ to $2.29\times$. The advantage is therefore consistent across the entire tested parameter space.

These results show that residual reuse is an important component of the algorithm's computational efficiency. Rather than repeatedly solving the green, green--brown, and green--brown--red prefix problems from scratch, \textsc{MaxFlow-Passes} directly extends the solution computed during the preceding phase.

\paragraph{Behavior across phases.}
% \label{sec:exp-flow-phases}

\begin{figure}[t]
    \centering
    \includegraphics[width=0.75\columnwidth]
        {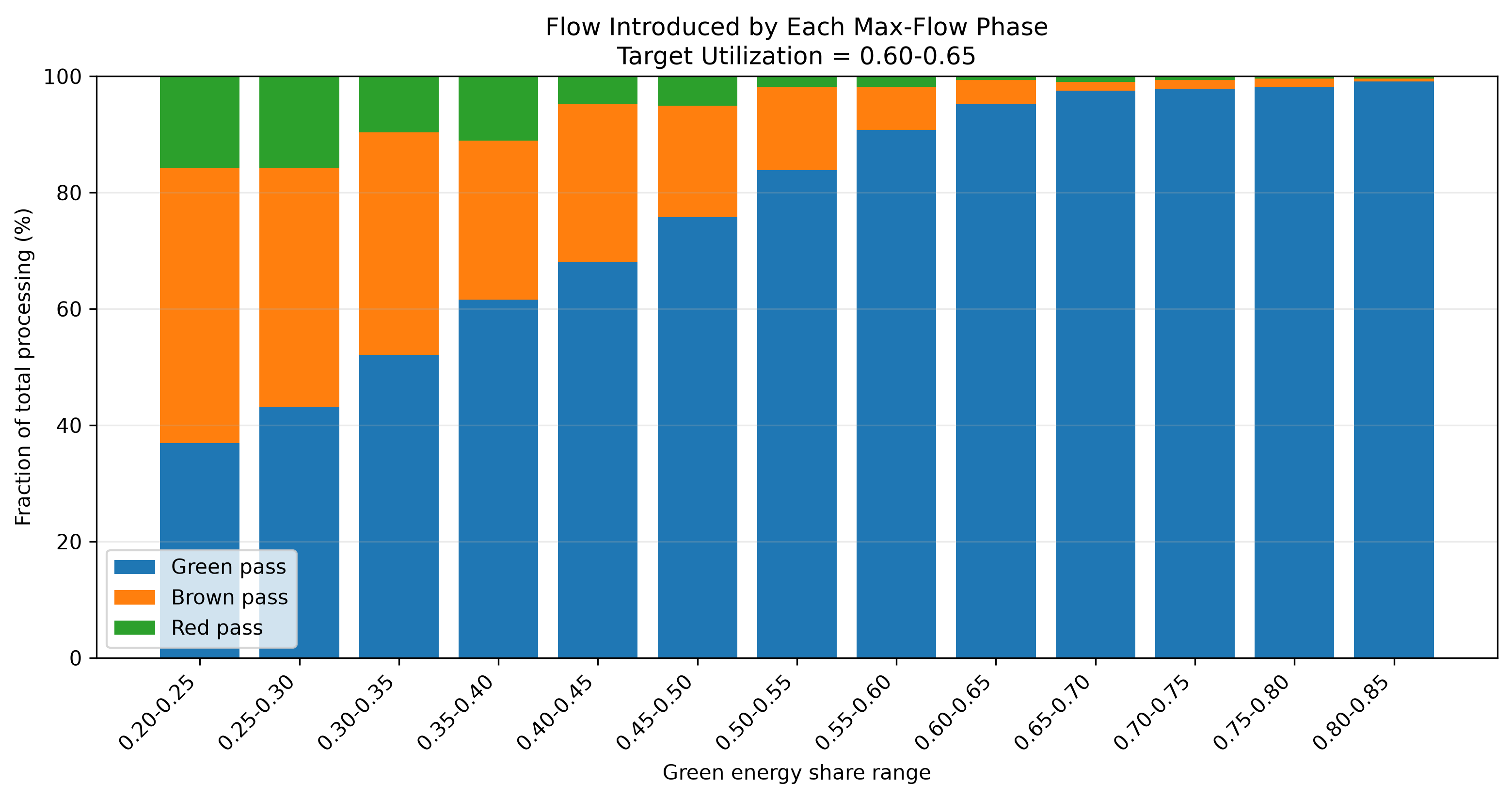}
    \caption{Fraction of total processing introduced by the green, brown,
    and red phases at target utilization $0.60$--$0.65$.  Increasing the
    green-energy share shifts progressively more of the final flow into
    the first phase.}
    \label{fig:flow-by-pass}
\end{figure}

Across all instances, the green phase contributes $73.85\%$ of the final flow on average. Brown capacity contributes another $18.14\%$, while red capacity contributes $8.01\%$. The corresponding median values are $79.10\%$, $11.84\%$, and $0\%$. Thus, in at least half of the tested instances, no additional processing is introduced during the red phase.

Figure~\ref{fig:flow-by-pass} shows this behavior for target utilization
$0.60$--$0.65$. At low green shares, a substantial fraction of the processing requirement is supplied by the brown and red phases. As the green share increases, more processing is established during the first phase. At the largest green shares, almost the entire final flow is present after the green phase.

This behavior helps explain the effectiveness of incremental
augmentation.  In many instances, the first phase establishes most of the
final solution, and later phases need only supply the remaining demand
rather than recomputing the entire allocation.

\begin{table*}
\caption{Aggregate behavior of the exact solvers and incremental max-flow phases.}
\centering
\label{tab:pass-mechanism-summary}
\begin{tabular}{lrrrrr}
\toprule
Metric & Mean & Median & Std. & Min & Max \\
\midrule
MinCost / MaxFlow-Passes runtime & 4.961 & 4.880 & 0.979 & 2.314 & 8.953 \\
Restart / MaxFlow-Passes runtime & 2.069 & 2.009 & 0.329 & 1.395 & 2.830 \\
\midrule
Flow added in green pass (\%) & 73.847 & 79.104 & 23.000 & 20.896 & 100.000 \\
Flow added in brown pass (\%) & 18.140 & 11.842 & 18.491 & 0.000 & 78.358 \\
Flow added in red pass (\%) & 8.013 & 0.000 & 13.722 & 0.000 & 76.226 \\
\midrule
Green allocation rerouted in Pass 2 (\%) & 3.626 & 0.000 & 6.930 & 0.000 & 65.079 \\
Prior allocation rerouted in Pass 3 (\%) & 1.595 & 0.000 & 4.075 & 0.000 & 36.441 \\
\bottomrule
\end{tabular}
\end{table*}

\subsection{Scalability}
\label{sec:exp-scalability}

We next study how \textsc{MaxFlow-Passes} scales as the scheduling instance and the corresponding flow network grow. Unlike the previous experiment, this study varies both the number of jobs and the number of carbon intervals over several orders of magnitude.

For each job count $n$, we consider five carbon-interval densities. The number of original intervals is approximately $0.10n$, $0.25n$, $0.50n$, $1.0n$, and $2.0n$. We refer to these settings as \emph{sparse}, \emph{light}, \emph{medium}, \emph{dense}, and \emph{very dense}. We generate three instances for each configuration and report the mean end-to-end running time.

The number of original carbon intervals does not directly determine the final network size.
Splitting the timeline at job release times and deadlines produces a larger set of atomic intervals. A job--interval edge is created whenever an atomic interval lies inside the feasible window of a job. We therefore also record
\[
    |E_{JI}|
    =
    \left|
    \left\{
        (j,k):
        I_k \subseteq [r_j,d_j)
    \right\}
    \right|.
\]

\paragraph{CPU scaling.}

% \begin{figure}[t]
%     \centering
%     \includegraphics[width=0.75\columnwidth]
%         {figures/existing_cpu_runtime_vs_jobs.pdf}
%     \caption{CPU running time of \textsc{MaxFlow-Passes} as the number
%     of jobs increases.  Each curve corresponds to a different number of
%     original carbon intervals relative to the job count.  Both axes are
%     logarithmic.}
%     \label{fig:cpu-scaling-jobs}
% \end{figure}

% \begin{figure}[t]
%     \centering
%     \includegraphics[width=0.75\columnwidth]
%         {figures/existing_cpu_runtime_vs_edges.pdf}
%     \caption{CPU running time of \textsc{MaxFlow-Passes} as a function
%     of the number of job--interval edges in the constructed flow network.
%     Each point is averaged over the available instances of the
%     corresponding configuration.  Both axes are logarithmic.}
%     \label{fig:cpu-scaling-edges}
% \end{figure}

\begin{figure}[t]
    \centering

    \begin{subfigure}[t]{0.48\columnwidth}
        \centering
        \includegraphics[width=\linewidth]
            {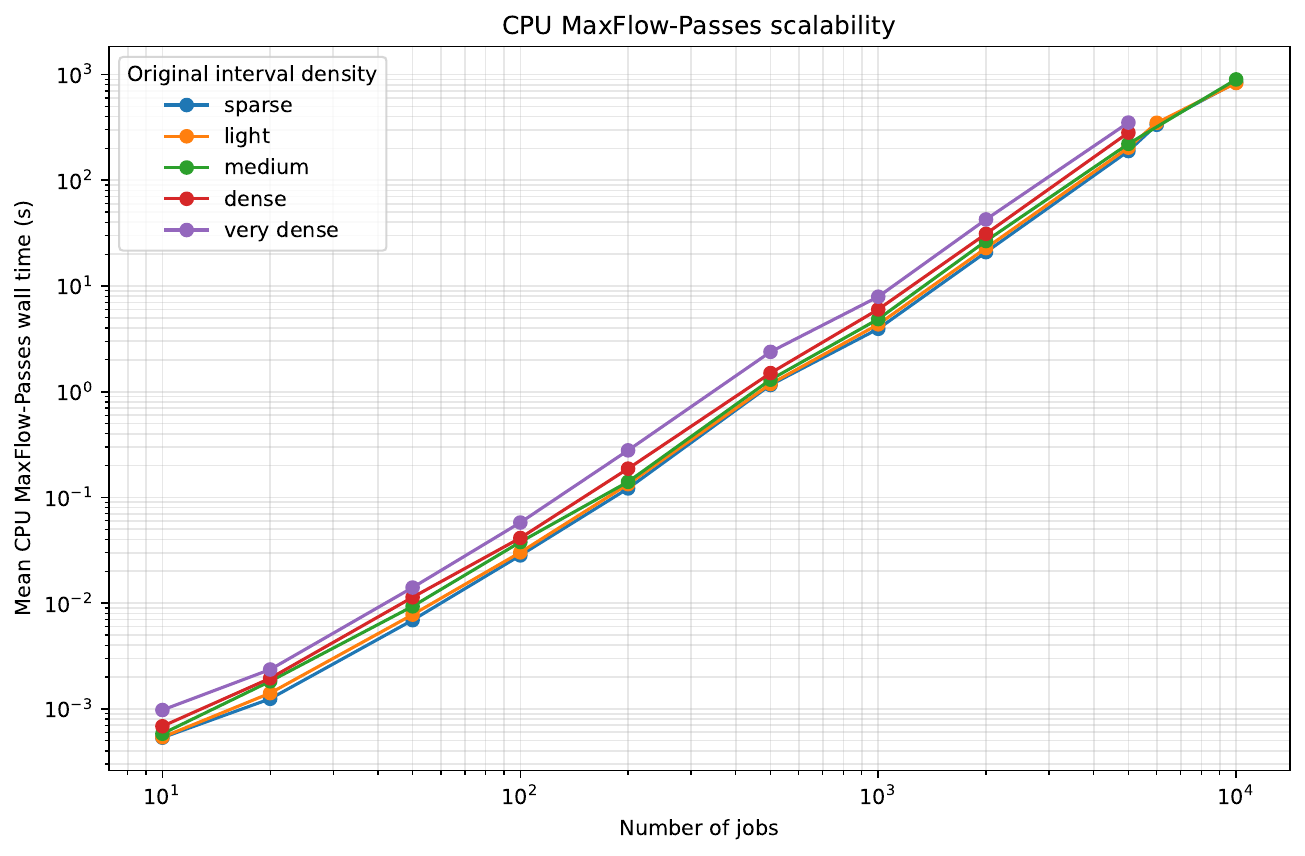}
        \caption{Number of jobs.}
        \label{fig:cpu-scaling-jobs}
    \end{subfigure}
    \hfill
    \begin{subfigure}[t]{0.48\columnwidth}
        \centering
        \includegraphics[width=\linewidth]
            {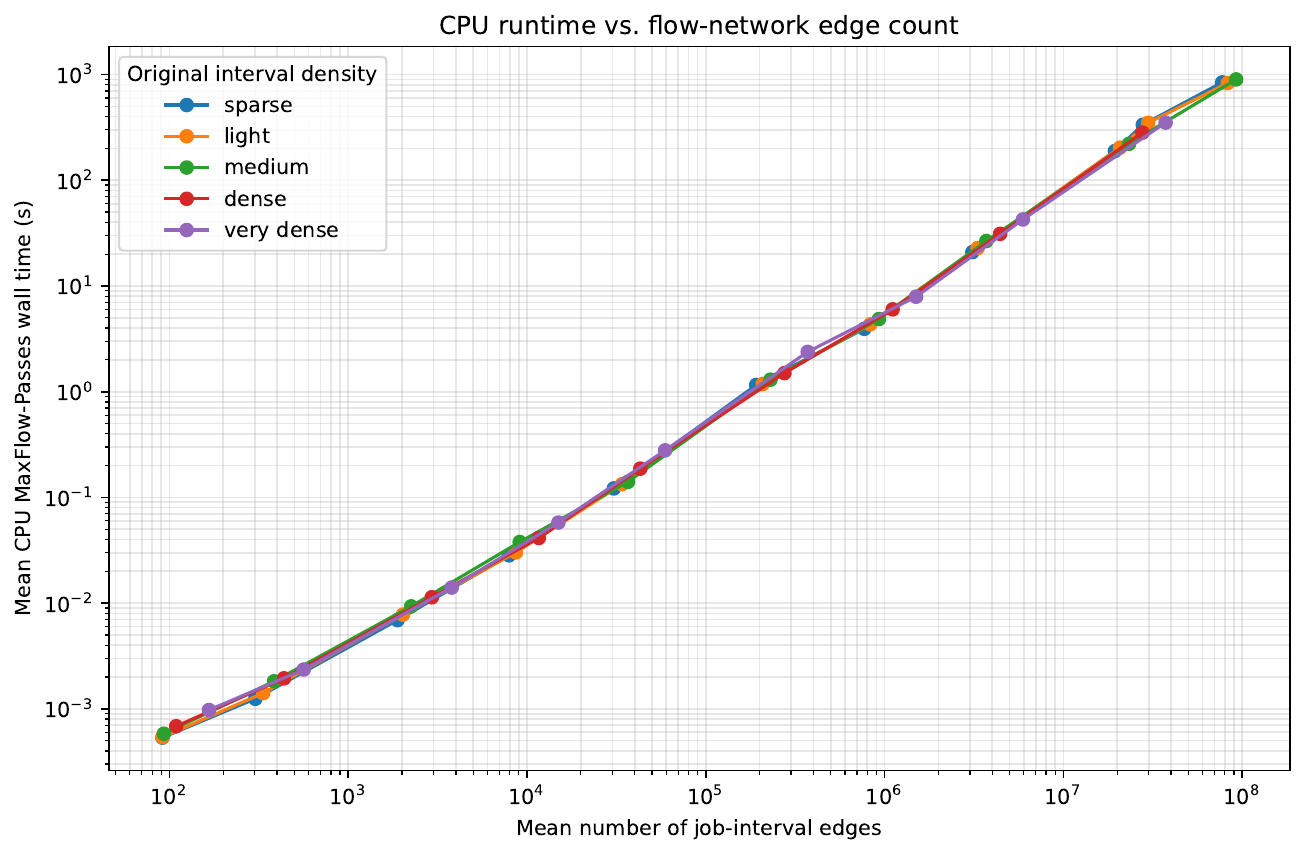}
       \caption{Job--interval edges.}
        \label{fig:cpu-scaling-edges}
    \end{subfigure}

    \caption{CPU running time of \textsc{MaxFlow-Passes}.
    (a) Runtime versus the number of jobs, with curves
    corresponding to different carbon-interval densities.
    (b) Runtime versus the number of job--interval edges
    in the constructed flow network, averaged over
    instances of each configuration.
    Both axes are logarithmic.}
    \label{fig:cpu-scaling}
\end{figure}

Figure~\ref{fig:cpu-scaling-jobs} shows that running time increases rapidly with the number of jobs.
For the sparse configuration, the mean running time increases from
$0.53$\,ms for $10$ jobs to $28.3$\,ms for $100$ jobs,
$3.93$\,s for $1000$ jobs, and $188.9$\,s for $5000$ jobs.
For the largest completed sparse configuration with $10{,}000$ jobs,
the mean running time is $843.0$\,s.

The flow network grows much faster than the job count alone. For the same sparse configurations, the mean number of job--interval
edges increases from $91$ at $10$ jobs to approximately
$7.9\times 10^3$ at $100$ jobs,
$7.7\times 10^5$ at $1000$ jobs, and
$1.94\times 10^7$ at $5000$ jobs.
At $10{,}000$ jobs, the network contains approximately
$7.74\times 10^7$ job--interval edges.

Interval density produces a similar effect. With $1000$ jobs, the mean number of job--interval edges grows from approximately $0.77$ million in the sparse setting to $1.50$ million in the very-dense setting. The corresponding mean running time increases from $3.93$,s to $7.91$,s. At $5000$ jobs, the network grows from approximately $19.4$ million edges in the sparse setting to $37.3$ million in the very-dense setting. Running time increases from $188.9$\,s to $351.2$\,s.

These results show that $|E_{JI}|$ is a useful measure of computational scale. Increasing the number of jobs also increases the number of atomic intervals and feasible job--interval pairs. Instances with the same job count can therefore produce substantially different network sizes and running times.

The implementation handles networks containing tens of millions of job--interval edges. For the largest instances, the running time reaches several minutes.

\paragraph{GPU acceleration.}
% \label{sec:exp-gpu}

We also evaluate whether \textsc{MaxFlow-Passes} benefits from GPU
acceleration. We implement \textsc{MaxFlow-Passes} using ECL-MaxFlow~\cite{ECLMaxFlow}.  The GPU version
preserves the flow between phases, but reconstructs the residual graph
and invokes the GPU solver separately for each pass.  
We report both GPU algorithm time and complete GPU wall-clock time. GPU algorithm time excludes ECL interface overhead.

\begin{figure}[t]
    \centering
    \includegraphics[width=0.98\columnwidth]
        {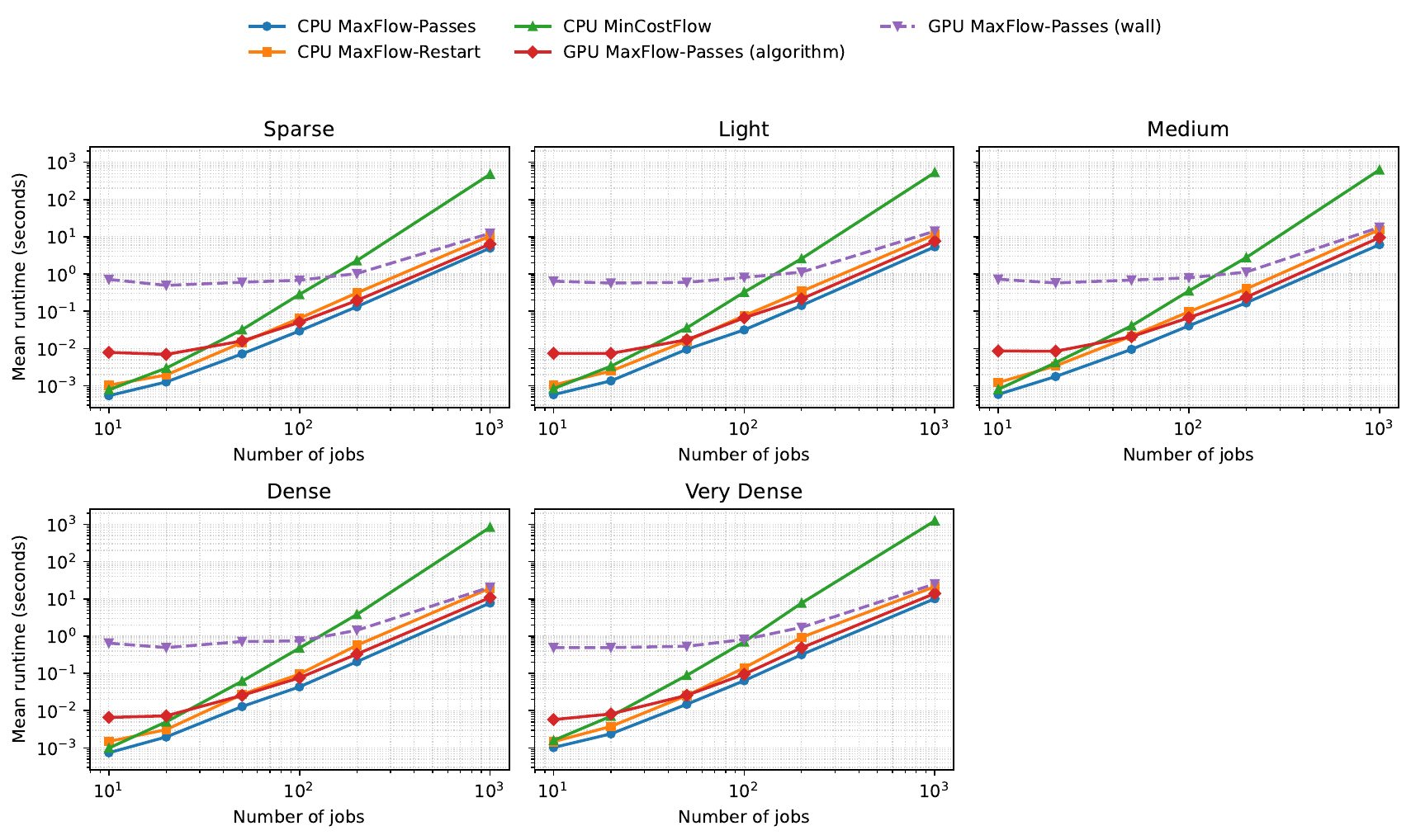}
    \caption{Running time of the CPU exact solvers and GPU
    \textsc{MaxFlow-Passes}. Both axes are logarithmic.}
    \label{fig:gpu-cpu-scaling}
\end{figure}

Figure~\ref{fig:gpu-cpu-scaling} shows that the GPU implementation does
not outperform CPU \textsc{MaxFlow-Passes} in the tested configurations.
For example, at $1000$ sparse jobs, CPU \textsc{MaxFlow-Passes} requires
$4.96$\,s, compared with $6.34$\,s of GPU algorithm time and
$12.35$\,s of GPU wall time.  Across the $1000$-job configurations, GPU
algorithm time is about $1.28\times$--$1.53\times$ the CPU
\textsc{MaxFlow-Passes} time.  However, it remains much faster than
\textsc{MinCostFlow}, whose runtime reaches hundreds of seconds at this
scale.

The current GPU integration therefore does not provide a speedup over the incremental CPU implementation. Residual construction, flow recovery, and repeated GPU invocation remain significant costs.

\subsection{Evaluation on Alibaba Cluster Traces}
\label{sec:exp-alibaba}

We next evaluate \textsc{MaxFlow-Passes} and \textsc{MinCostFlow} on
production-derived workloads from the Alibaba cluster trace.  Each selected Alibaba job is converted into one scheduling instance. Its tasks form the precedence DAG.

For each task $j$: release time $r_j=\min_k s_{jk}$, deadline $d_j=\max_k e_{jk}$, and weight $w_j=\sum_k(e_{jk}-s_{jk})c_{jk}$ are set, where $s_{jk}$, $e_{jk}$, and $c_{jk}$ are the start time, end time, and
average CPU utilization of instance $k$.  Processing times are scaled as
$p_j=\alpha w_j$, where $\alpha$ is chosen so that
$\sum_j p_j=0.7H,$ and $H$ is the observed job horizon.
Release times and
deadlines are then tightened according to the DAG precedence constraints.

The horizon is divided into six equal carbon intervals following
green--brown--red--green--brown--red, with costs $0$, $1$, and $2$.

We select valid trace jobs over increasing DAG sizes and report their actual
number of tasks.  Each solver is run five times and we report median
end-to-end runtime.  We additionally record the number of job--interval
edges $|E_{JI}|$, since this more directly characterizes the size of the
constructed flow network.

\begin{figure}[t]
    \centering
    \includegraphics[width=0.75\columnwidth]
        {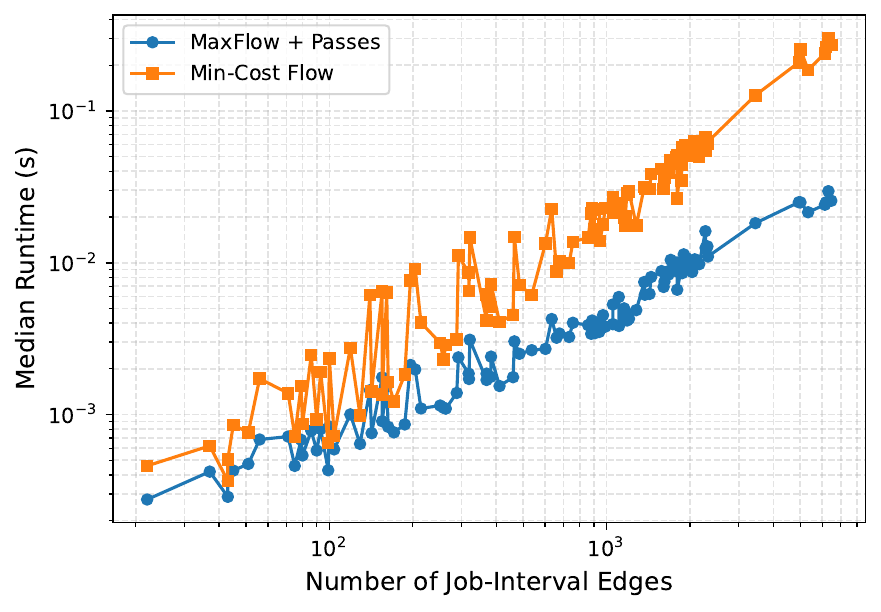}
    \caption{Running time of the exact solvers on Alibaba trace-derived
    workloads as the constructed scheduling instances increase in size.}
    \label{fig:alibaba-runtime}
\end{figure}

Figure~\ref{fig:alibaba-runtime} shows that
\textsc{MaxFlow-Passes} remains consistently faster than
\textsc{MinCostFlow} on the trace-derived workloads.  The two algorithms
agree on feasibility and optimal carbon cost for every tested instance.
The runtime gap generally increases for larger flow networks, reaching
approximately $10\times$ on some of the largest networks. 

The advantage observed on synthetic instances therefore also appears on the production-derived workloads.

\subsection{Cloud-Outsourcing Reduction and Lawler's Algorithm}
\label{sec:exp-lawler}

% We evaluate both the correctness and scalability of the reduction in
% Section~\ref{sec:cloud}. For validation, we generate 30 instances for
% each $n\in\{5,8,10\}$ and assign outsourcing penalties uniformly from
% $\{1,\ldots,10\}$. We compute the cloud optimum independently by
% enumerating all local/cloud assignments and solving each local
% scheduling problem using \textsc{MinCostFlow}. We compare this value
% with the penalty optimum obtained from Lawler's exact dynamic program.
% The two objective values match on all 90 instances.

We next evaluate the scalability of Lawler's exact algorithm and the
resulting edge/cloud allocation. We vary the number of original jobs
over $n\in\{5,10,15,20,25,30,35,40\}$ and the number of carbon
intervals over approximately $0.10n$, $0.25n$, $0.50n$, and $1.0n$,
denoted \emph{sparse}, \emph{light}, \emph{medium}, and \emph{dense},
respectively. Outsourcing penalties are sampled
uniformly from $\{1,\ldots,\kappa_{\max}\}$, where
$\kappa_{\max}\in\{5,10,20,40,80\}$. We generate five instances per
configuration, yielding 800 instances in total, using random seed
$42$. Each instance has a 1000-second time limit.

For each reduced instance, we measure the running time, total weight
$W$, and number of original jobs assigned to the cloud in an optimal
solution. Figure~\ref{fig:lawler-results} illustrates the computational
behavior of Lawler's algorithm. Figure~\ref{fig:lawler-scaling} shows
that runtime increases rapidly with the number of original
jobs, with denser carbon intervals generally requiring more
computation. This is partly due to the additional dummy jobs
introduced by the reduction. Figure~\ref{fig:lawler-weight-scaling} examines
the effect of penalty weights by scaling all weights by a
common factor $\alpha$, while keeping the scheduling
constraints unchanged. Increasing $\alpha$ increases the
total weight $W$ without changing the optimal set of
accepted jobs, allowing us to isolate the effect of weight
magnitude on the dynamic program's runtime.

Across the evaluated configurations, the percentage of original
jobs outsourced to the cloud ranges from approximately $2\%$
to $27\%$. This shows that cloud outsourcing plays a
nontrivial role in the optimal solutions, with the allocation
between edge and cloud varying across different scheduling
configurations and outsourcing penalties.

\begin{figure}[t]
    \centering

    \begin{subfigure}[t]{0.48\columnwidth}
        \centering
        \includegraphics[width=\linewidth]
            {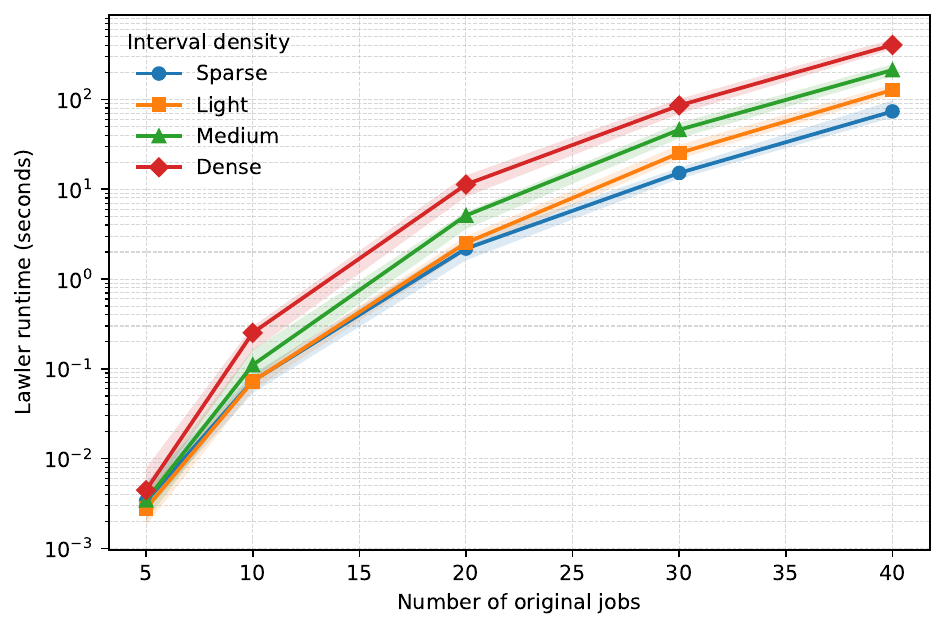}
        \caption{Number of jobs.}
        \label{fig:lawler-scaling}
    \end{subfigure}
    \hfill
    \begin{subfigure}[t]{0.48\columnwidth}
        \centering
        \includegraphics[width=\linewidth]
            {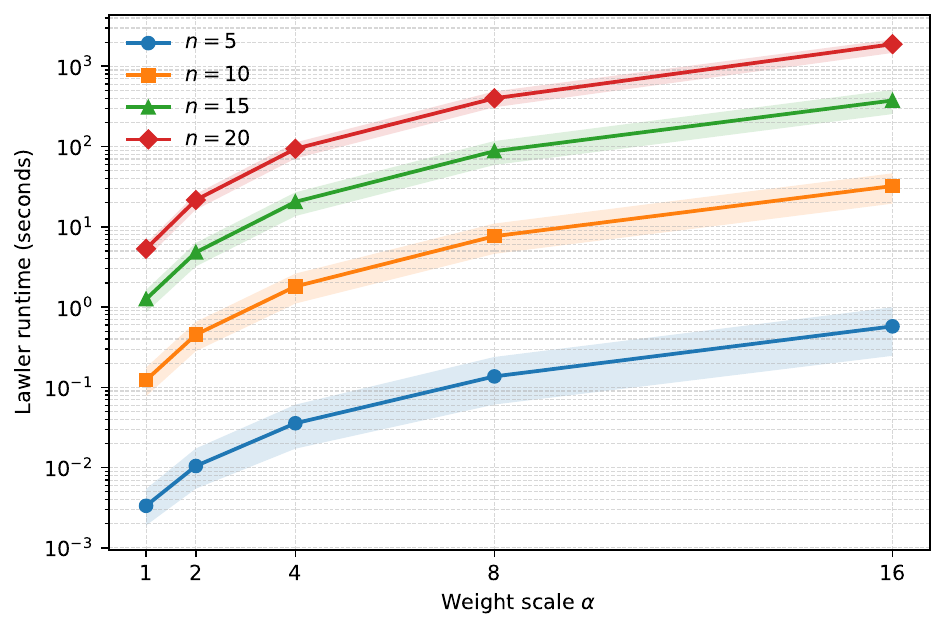}
        \caption{Total weight.}
        \label{fig:lawler-weight-scaling}
    \end{subfigure}

    \caption{Running time of Lawler's exact dynamic program.
    Shaded regions show the interquartile range; the $y$-axis is
    logarithmic.}
    \label{fig:lawler-results}
\end{figure}

\section{Conclusions}

In this paper, we developed optimal strategies for pre-emptive scheduling of jobs on a single edge server with time varying costs, under general precedence constraint requirements, with release times and deadlines. Moreover, we showed that network flow methods are a faster and easier way to find optimal solutions instead of using the more general min-cost flow based methods. In addition, we developed algorithms based on pioneering work by Lawler \cite{Lawler} for cloud out-sourcing - this connects the edge vs cloud question for the first time with classical work on weighted tardiness scheduling. However, the algorithm is not yet scalable. Meanwhile, what would be interesting to study is fast approximation algorithms for minimizing weighted tardiness, a problem now of central interest in the context of making decisions as to what to outsource to the more expensive cloud computing platform from the edge resource.

\bibliographystyle{abbrv}
\bibliography{refs}

\end{document}